\documentclass[lettersize,journal]{IEEEtran}

\usepackage{amsmath,amssymb,amsthm}
\usepackage{graphicx}
\usepackage{booktabs}
\usepackage{multirow}
\usepackage{algorithm}
\usepackage[noend]{algorithmic}
\usepackage{xcolor}
\usepackage[hidelinks]{hyperref}
\usepackage{cite}
\usepackage{array}
\usepackage{balance}
\usepackage{subcaption}

\newtheorem{theorem}{Theorem}
\newtheorem{definition}{Definition}
\newtheorem{proposition}{Proposition}

\newtheorem{corollary}{Corollary}

\newcommand{\store}{\mathcal{M}}
\newcommand{\query}{q}
\newcommand{\memory}{m}
\newcommand{\key}{K}

\newcommand{\ASR}{\mathrm{ASR}}

\begin{document}

\title{SMSR: Certified Defence Against Runtime Memory Poisoning\\
       in Persistent LLM Agent Systems}

\author{Tarun~Sharma%
        \thanks{T.~Sharma is an Independent Researcher
        (e-mail: tarun.sharma@ieee.org).}%
        \thanks{Code and data are available at \url{https://github.com/tarun-ks/smsr}.}%
        \thanks{This work has been submitted to the IEEE for possible publication.
        Copyright may be transferred without notice, after which this version may
        no longer be accessible.}}

\markboth{Preprint}%
{Sharma: SMSR: Certified Defence Against Runtime Memory Poisoning}

\maketitle

\begin{abstract}
Retrieval-Augmented Generation (RAG) systems underpin an increasing fraction
of enterprise AI deployments.  When an agent's memory store persists across
user sessions, an adversary who can interact with the system can inject
carefully crafted memories that, once retrieved, redirect the agent's
behaviour on future queries---without ever modifying model weights or
application code.  We call this the \emph{Multi-Session Memory Poisoning}
(MSMP) threat and observe that no existing defence provides a formal security
certificate against it.  Static-corpus defences (RobustRAG, ReliabilityRAG)
assume a fixed knowledge base; heuristic filters are bypassable by fluent
enterprise-style text.

We present \emph{Signed Memory with Smoothed Retrieval} (SMSR), a two-component
defence that is the first to provide a certified robustness bound for the MSMP
setting.  Component~1 applies HMAC-SHA256 provenance tagging at write time,
creating a cryptographically hard boundary against unsigned injection.
Component~2 applies randomised memory ablation at query time with a
verdict-based majority aggregator, bounding the influence of authenticated
adversaries.

We prove an impossibility result showing that no provenance-free retrieval-time
filter can certify against adaptive injection, and derive a hypergeometric
certificate for Component~2.  We also formalise and quantify the
\emph{Consistent Minority Effect} in the memory-poisoning setting: string-based
majority vote is gamed by adversaries who generate consistent responses, and
verdict-based aggregation removes the effect.

Empirical evaluation on 15 enterprise knowledge-base scenarios (3{,}150 repeated
trials: six Component-2 configurations $\times$ 15 scenarios $\times$ 30
repetitions, plus 450 production-scale trials) reports five findings.
\textbf{First}, Component~1 reduces ASR from 93--100\% to \textbf{0\%} for all
unsigned injection variants including bypass attacks crafted to evade heuristic
filters.
\textbf{Second}, in a production-scale store ($m=20$, 20 seed memories,
Tier~1), Component~2 reduces authenticated ASR from 93--100\% to \textbf{8.0\%}
(95\% CI [5.8\%, 10.9\%], $n=450$) for $t=1$, which lies safely below the
Theorem-2 worst-case bound $\delta=10.4\%$ at our evaluated setting
$n_\text{runs}=5$ (the bound tightens to 7.1\% at $n_\text{runs}=7$, which we
do not evaluate).
\textbf{Third}, in our small-store evaluation ($m'=10{+}t$), the canonical
direct-injection $t=1$ ASR is \textbf{37.8\%} (95\% CI [33.4\%, 42.3\%],
$n=450$), below the eval-pool bound $\delta=41.5\%$; the direct-injection
bound holds at $t\in\{1,2,3\}$, while the flooding variant sits at the
worst-case bound (its CI straddles $\delta$), as expected for a tight
certificate.
\textbf{Fourth}, judge reliability is $\kappa=0.955$ (Haiku vs Sonnet, $n=84$);
in the production-scale 20-seed store, Sonnet and Haiku both achieve 8.0\% ASR
under Component~2, confirming the defence generalises across model families
(Section~\ref{sec:e7auth}).
\textbf{Fifth}, in an end-to-end test where the agent itself writes the poison
via a query-only attack (not pre-seeded), SMSR reduces ASR from 65.3\% to 5.3\%
($n=150$, non-overlapping CIs; Section~\ref{sec:e10}).
Utility on clean queries is 90\% under Component~1 and 85\% under the combined
defence.
\end{abstract}


\section{Introduction}
\label{sec:intro}

Retrieval-Augmented Generation (RAG) agents are deployed at scale in
enterprise environments for internal search, automated customer support,
compliance checking, and agentic process automation~\cite{owasp2025,csa2026}.
A defining characteristic of production deployments is \emph{persistent
memory}: the agent accumulates a growing store of past interactions, retrieved
documents, and inferred facts across user sessions, enabling context continuity
and personalisation.

\textbf{The attack surface.}
This persistence creates a novel attack vector: an adversary who can interact
with the system through normal channels---whether as an employee, a customer, or
a compromised upstream tool---can craft a sequence of interactions that, once
stored in the memory bank, persist indefinitely and corrupt future responses for
any user whose query semantically matches the poisoned entry.  Unlike training
data poisoning or jailbreak attacks that target model weights or prompts
directly, this \emph{runtime memory poisoning} operates at the retrieval layer
and leaves no trace in the model itself.

MINJA~\cite{minja2025} demonstrated that this attack achieves $\approx$76--99\%
success rates against memory-augmented agents.  AgentPoison~\cite{agentpoison2024}
showed that a single poisoned entry in a 23k-entry database suffices for 62\%
end-to-end attack success.  MemoryGraft~\cite{memgraft2025} demonstrated
persistent compromise across sessions at 48\% poisoned retrieval proportion.
None of the proposed countermeasures in these works provides a formal
security guarantee.

\textbf{Why existing defences are insufficient.}
The two most rigorous existing defences, RobustRAG~\cite{robustrag2024} and
ReliabilityRAG~\cite{reliabilityrag2025}, assume a \emph{static} document
corpus that the adversary poisons externally before queries are made.  They
provide no security guarantee when poisoned entries are injected at runtime
through the agent's normal memory write path.  Heuristic defences such as A-MemGuard~\cite{amemguard2025} use consensus
validation across multiple reasoning paths.  In our like-for-like evaluation
(Section~\ref{sec:eval}), A-MemGuard attains empirical ASR comparable to SMSR
at $t=1$, but it provides \emph{no formal guarantee}: SMSR's distinction is a
certified robustness bound combined with over-fetch random sampling that,
unlike full-retrieval consensus, degrades gracefully as a persistent adversary
adds entries.  As a general-purpose heuristic baseline we implement a
keyword blacklist, entropy proxy, and semantic anomaly filter (not attributed
to any specific system) and show empirically that all three are fully bypassed
(100\% ASR) by fluent enterprise-style policy text, motivating the need for
the provenance mechanism in Component~1.

\textbf{Contributions.}
This paper makes the following contributions:
\begin{enumerate}
  \item \textbf{Formal definitions}: We introduce the Multi-Session Memory
    Poisoning (MSMP) threat model and the first formal $(t,\delta)$-security
    definition for runtime agent memory systems (Section~\ref{sec:model}).

  \item \textbf{Impossibility result}: We prove that no provenance-free
    retrieval-time filter can achieve a non-trivial security certificate against
    an adaptive MSMP adversary (Theorem~\ref{thm:impossibility}).

  \item \textbf{SMSR construction}: We present the Signed Memory with
    Smoothed Retrieval defence, combining HMAC provenance (Component~1) with
    randomised ablation and verdict-based majority aggregation (Component~2),
    and derive a formal hypergeometric certificate for Component~2
    (Theorem~\ref{thm:certificate}).

  \item \textbf{Consistent Minority Effect (CME)}: We formalise and
    characterise the CME --- a known self-consistency pitfall~\cite{wang2023selfconsistency}
    --- in the memory-poisoning setting: an adversarial response wins string-based
    majority vote despite being a numerical minority, because it is textually
    more consistent than the varied benign responses.  We show that
    verdict-based aggregation removes the effect and quantify it honestly
    (Proposition~\ref{prop:cma}, Section~\ref{sec:cma}).

  \item \textbf{Empirical validation}: We evaluate SMSR on 15 enterprise
    knowledge-base scenarios across three attack classes (unsigned injection,
    authenticated injection, heuristic bypass) and four defence configurations,
    using a LLM judge for reliable verdict evaluation.
\end{enumerate}

\begin{figure*}[t]
\centering
\includegraphics[width=0.96\textwidth]{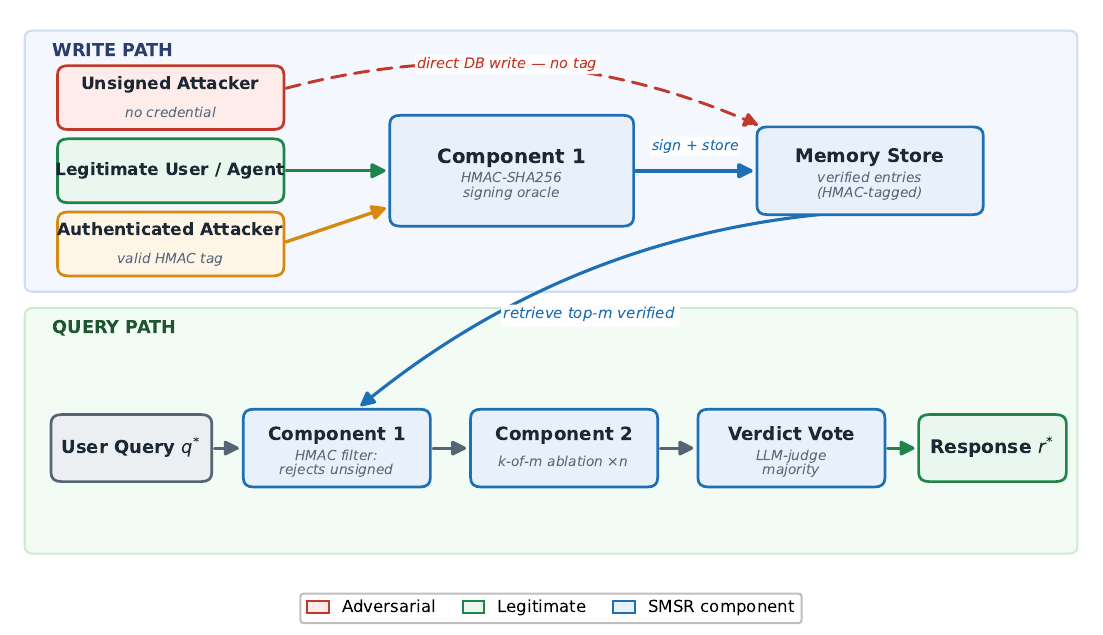}
\caption{SMSR system architecture. The write path (top) shows how the
HMAC signing oracle tags every legitimate memory and how an unsigned attacker
is blocked at retrieval time by Component~1.  An authenticated attacker
(legitimate user) can write signed memories but is mitigated by the randomised
ablation and verdict-based aggregation of Component~2.  The query path (bottom)
shows both components in sequence.}
\label{fig:arch}
\end{figure*}

\section{Background and Related Work}
\label{sec:related}

\subsection{Memory-Augmented LLM Agents}

Modern agentic systems such as LangGraph~\cite{langgraph}, AutoGen~\cite{autogen},
and MemGPT maintain persistent memory as a vector database over agent
interaction traces.  At query time the agent embeds the user's query
$\phi(\query) \in \mathbb{R}^d$ and retrieves the $k$ entries from the
memory store with highest cosine similarity.  These entries are prepended
to the LLM's context as few-shot demonstrations or factual context,
directly influencing the response.

Unlike static RAG over a curated document corpus (e.g., a company policy PDF
index), agent memory is continuously written to by the agent itself during
normal operation---every interaction typically generates one or more new
memory entries.  This creates the fundamental tension at the heart of the MSMP
threat: the write path that enables the agent's utility (accumulating useful
context over time) is also the injection path for an adversary.
Traditional access-control solutions that distinguish ``data'' from ``instructions''
do not apply, because in an agent memory system, data IS instructions---retrieved
memories function as in-context demonstrations that shape future behaviour.

Multi-tenant deployments exacerbate the problem: a shared memory store serving
multiple users means one user's interactions can affect all future users' responses.
This is the precise scenario exploited by MINJA~\cite{minja2025}, which achieves
$\approx$98\% injection success rate across GPT-4/Claude in shared agent settings.

\subsection{The OWASP LLM08 Threat}

The Open Web Application Security Project (OWASP) recognises ``Vector and
Embedding Weaknesses'' as LLM08 in its 2025 Top 10 for LLM
Applications~\cite{owasp2025}.  LLM08 covers both \emph{inversion} attacks
(recovering original text from stored embeddings) and \emph{poisoning} attacks
(injecting malicious content into the retrieval pipeline).  Our work addresses
the poisoning sub-problem, specifically the dynamic injection variant that
is absent from all prior formal treatments of LLM08 defences.

\subsection{Memory Poisoning Attacks}

MINJA~\cite{minja2025} is the first systematic study of runtime memory
injection: the adversary sends crafted queries that cause the agent to generate
and store poisoned reasoning traces, which are subsequently retrieved as
few-shot demonstrations for future users ($\approx$98\% injection success rate).
AgentPoison~\cite{agentpoison2024} targets RAG knowledge bases with a trigger
token that clusters poisoned entries in embedding space, achieving 62.6\%
end-to-end attack success.  MemoryGraft~\cite{memgraft2025} poisons an
agent's persistent experience store via a benign-looking artefact, demonstrating
that 10 poisoned seeds in a 110-entry store yield 48\% poisoned retrieval.
The Agent Security Bench~\cite{agentsecuritybench} catalogues these and related
agent attacks in a unified benchmark.

None of these attack papers proposes a certified defence; their countermeasure
discussions are heuristic.  Concurrent work by Devarangadi~Sunil et
al.~\cite{memory_poisoning_defense2025} proposes trust-scoring and
memory-sanitisation defences and independently observes that a store
pre-populated with legitimate memories dilutes attack success---the same
effect our Component-2 certificate (Theorem~\ref{thm:certificate}) quantifies
formally---but their defences remain heuristic with no certified bound.

\subsection{Defences for RAG Systems}

\textbf{Why static-corpus defences fail for MSMP.}
RobustRAG~\cite{robustrag2024} and ReliabilityRAG~\cite{reliabilityrag2025}
provide certified bounds under a model that is structurally incompatible with
the MSMP setting.  Both assume: (1) the corpus is indexed offline before
queries are answered; (2) the adversary poisons the corpus externally; and
(3) the number of poisoned documents in the retrieved set is an \emph{a priori}
known constant.  In the MSMP setting all three assumptions are violated: the
memory store is a live append-only log, the adversary operates through the
same write path as legitimate users, and the number of adversarial entries in
the retrieved set is unknown and depends on adversary persistence.

The ``isolate-then-aggregate'' strategy of RobustRAG requires partitioning
the retrieved set into $\lceil k/k' \rceil$ disjoint groups, each processed
independently.  A persistent adversary in a dynamic memory system can plant
entries in every partition, collapsing the certificate.  Our Component~2
addresses this via over-fetch randomisation ($m \gg k$) rather than deterministic
partitioning, preventing the adversary from guaranteeing presence in every
sampled context.

\textbf{Overview.}
RobustRAG~\cite{robustrag2024} applies an isolate-then-aggregate strategy
that provides a certified lower bound on response quality when at most $k'$
of the $k$ retrieved passages are adversarially injected.  ReliabilityRAG
\cite{reliabilityrag2025} builds on this with a Maximum Independent Set
construction on a document contradiction graph.  Both require a fixed,
pre-indexed corpus and explicitly do not address dynamic memory writes; their
proofs rely on a static adversary model.

For the query privacy dimension of RAG, Boldyreva and Tang~\cite{boldyreva2021}
formalise simulation-based security for private approximate $k$-NN search,
covering access, query, and volume pattern leakage.  The SoK by Bodea
et al.~\cite{sok2026} systematically reviews the RAG \emph{privacy} literature;
no certified poisoning defence for the dynamic/runtime memory setting is
known to us --- the closest works treat static-corpus robustness~\cite{robustrag2024,reliabilityrag2025}
and leave the dynamic setting open.

\subsection{Randomised Smoothing and Ablation}

Cohen et al.~\cite{cohen2019} introduced randomised smoothing to certify
$\ell_2$ robustness for image classifiers.  Levine and Feizi~\cite{levine2020}
adapted this to text by randomised ablation: a classifier trained on randomly
masked inputs inherits a certified bound based on the proportion of tokens
the adversary can control.  Our Component~2 follows the same structural argument from Levine and
Feizi~\cite{levine2020} (for image classifiers) and Zeng et al.~\cite{zeng2023certified}
(for text), adapted to the retrieval-over-memory setting where the adversary
controls a bounded number of entries in the retrieval candidate pool.

\section{Threat Model and Problem Formulation}
\label{sec:model}

\subsection{System Model}

We consider a multi-session RAG agent $\mathcal{A}$ with:
\begin{itemize}
  \item A persistent memory store $\store = \{m_1, \ldots, m_N\}$ where each
    entry $m_i$ is a text string and its embedding.  $\store$ is append-only
    during normal operation; entries are written by the agent after each
    interaction.
  \item A retrieval function $\textsc{Retrieve}(\query, \store, k)$ that
    returns the $k$ entries most similar to query $\query$ by cosine similarity.
  \item An LLM response generator $\mathcal{G}$ that takes query $\query$ and
    retrieved context $\mathcal{C}$ and produces a response $r$.
\end{itemize}

Sessions are independent in that a new user at session $s'$ has no knowledge
of prior session $s$, but both sessions query the same persistent $\store$.

\subsection{Adversary Model}

We consider two adversary classes of increasing capability:

\textbf{Unsigned adversary $\mathcal{A}_U$}: Has direct write access to
$\store$ (via database misconfiguration, SQL injection, or stolen backup
with append rights) and can insert entries without any authentication
credential.  Cannot modify existing entries or delete anything.

\textbf{Authenticated adversary $\mathcal{A}_P$}: Is a legitimate user of
the system; can write entries through the normal agent interaction path.
Can inject at most $t$ adversarially crafted entries per attack campaign.
Does not know the server-side secret key $\key$.

Both adversaries know the embedding model, the agent architecture, and the
retrieval mechanism.  Neither knows the HMAC key $\key$.

\subsection{Security Definitions}

\begin{definition}[Multi-Session Memory Poisoning (MSMP)]
\label{def:msmp}
An MSMP adversary $\mathcal{A}$ against agent $\mathcal{A}$ with memory
$\store$ is a probabilistic algorithm that:
\begin{enumerate}
  \item (Write phase) Injects at most $t$ memories $\{m^*_1,\ldots,m^*_t\}$
    into $\store$ via either the unsigned or authenticated write path.
  \item (Trigger phase) Observes that a future user issues a target query
    $\query^*$ semantically related to the injected content.
\end{enumerate}
The attack succeeds if $\mathcal{G}(\query^*, \textsc{Retrieve}(\query^*,
\store, k))$ produces a response that the adversary designates as
``malicious'' (a pre-specified false claim).
\end{definition}

\begin{definition}[$(t,\delta)$-SMSR Security]
\label{def:smsr_security}
A memory retrieval system is $(t,\delta)$-\emph{SMSR-secure} if for any MSMP
adversary $\mathcal{A}$ that injects at most $t$ adversarial entries, the
probability that the agent's response is malicious is at most $\delta$:
\[
  \Pr[\mathcal{A}(\query^*, \store) = \text{malicious}] \le \delta.
\]
\end{definition}

\section{Impossibility of Provenance-Free Certified Defence}
\label{sec:impossibility}

We first establish that any defence operating solely at retrieval time---without
write-time provenance---cannot achieve a non-trivial security certificate against
an adaptive adversary.

\begin{theorem}[Impossibility of Provenance-Free Worst-Case Certification]
\label{thm:impossibility}
Let $f: \mathcal{M} \times \mathcal{Q} \to \{0,1\}$ be any \emph{deterministic}
retrieval-time filter that decides whether to include memory $\memory$ in the
context for query $\query$ based solely on content.
Under the following standard embedding assumption:

\noindent\textbf{Assumption~A1} (\emph{fluent embedding density}):
For any target query embedding $e^*$ and threshold $\tau_\text{sim} > 0$,
there exists a fluent text string whose embedding has cosine similarity
$\ge \tau_\text{sim}$ with $e^*$.

\noindent then for any unbounded adaptive adversary $\mathcal{A}$, no $f$
achieves a non-trivial worst-case certificate: there always exists an
adversarial entry $m^*$ such that $f(m^*, \query^*)=1$.

\noindent \emph{Note:} Component~2 of SMSR is \textbf{not} an instance of the
filter class covered by this theorem.  Component~2 is a \emph{randomised,
content-agnostic} mechanism whose certificate (Theorem~2) relies on a
\emph{count bound} $t$ and \emph{independent sampling}, not on content
decisions.  The impossibility applies only to deterministic content-based
filters.
\end{theorem}

\begin{proof}
We construct $\mathcal{A}$ explicitly.  Let $\phi : \mathcal{M} \to \mathbb{R}^d$
be the embedding function.  $\mathcal{A}$ computes the target embedding
$e^* = \phi(\query^*)$ and constructs $m^*$ as a fluent text string that:
(1) is semantically on-topic for $\query^*$, so that $\langle \phi(m^*), e^* \rangle >
\tau_\text{sim}$ for any relevance threshold, and
(2) asserts a false factual claim in natural language.

The existence of such $m^*$ follows from two observations.  First, the
embedding space $\mathbb{R}^d$ contains infinitely many points within any
$\ell_2$ ball around $e^*$; by the density of fluent text in embedding space
(any region around a real document's embedding is populated by paraphrases),
there exists a fluent string $m^*$ with $\langle \phi(m^*), e^* \rangle > \tau_\text{sim}$.
Second, appending a false claim to an on-topic string changes its embedding by
at most $O(\|m^*\|^{-1})$ (since the false claim is a small fraction of a long
entry), preserving the similarity bound for sufficiently long $m^*$.

Since $f$ evaluates only $(\memory, \query)$ at retrieval time, and we have
constructed $m^*$ to pass any threshold-based criterion, $f(m^*, \query^*) = 1$
while $m^*$ causes the LLM to output the malicious claim.  Any $\delta < 1$
certificate for $f$ is therefore vacuous.
\end{proof}

This result implies that a provenance mechanism at \emph{write time} is
necessary.  Content-based filters (keyword blacklist, perplexity, semantic
anomaly) are instances of $f$ and are subject to this impossibility.

\section{The SMSR Construction}
\label{sec:construction}

\subsection{Component 1: HMAC Provenance Tagging}

MemoryGraft~\cite{memgraft2025} proposed ``cryptographic provenance
attestation'' as a defence direction without formalising or certifying it.
Component~1 formalises this direction: we prove its security under standard
HMAC assumptions (Proposition~\ref{prop:c1}) and establish it as the necessary
write-time boundary that enables Component~2's certificate.  Our contribution
is the formal proof and the integration with Component~2, not the concept.

Every legitimate memory write passes through a trusted server-side signing
oracle that tags each entry with an HMAC-SHA256 signature under a server
secret $\key$:
\[
  \tau_i = \textsc{HMAC}_\key(\text{content}_i \,\|\, \text{session\_id}_i
           \,\|\, \text{timestamp}_i).
\]
At retrieval time, only entries with valid tags are admitted to the candidate
pool:
\[
  \store_{\text{verified}} = \{\memory_i \in \store :
    \textsc{Verify}_\key(\memory_i) = 1\}.
\]

\begin{proposition}[Security of Component 1]
\label{prop:c1}
Under the pseudorandomness of HMAC-SHA256, an unsigned adversary $\mathcal{A}_U$
who does not know $\key$ produces a valid tag for any entry with probability
at most $2^{-\lambda}$, where $\lambda = 256$ is the output length in bits.
Therefore Component~1 achieves $(0, 2^{-256})$-SMSR-security against $\mathcal{A}_U$.
\end{proposition}

The server key $\key$ should be stored in a secrets manager or HSM, never
in application code.  Key rotation invalidates all pre-rotation memories,
which provides forward secrecy at the cost of losing history.

\subsection{Component 2: Randomised Memory Ablation with Verdict-Based Aggregation}

Component~1 handles unsigned adversaries completely.  For authenticated
adversaries $\mathcal{A}_P$ who can write HMAC-signed memories, a second
layer is required.

\textbf{Retrieval with ablation.}
Given a query $\query$, retrieve the top-$m$ verified candidates:
$\mathcal{C} = \textsc{Retrieve}(\query, \store_{\text{verified}}, m)$
where $m > k$.  For each of $n_\text{runs}$ independent trials, sample
$k$ entries uniformly at random without replacement from $\mathcal{C}$:
$\mathcal{C}^{(j)} \sim \binom{\mathcal{C}}{k}$.

\textbf{Verdict-based aggregation.}
Run the LLM with each context $\mathcal{C}^{(j)}$ to obtain response
$r^{(j)}$.  Apply a per-run judge $J(r^{(j)}, \query)$ that returns a verdict
$v^{(j)} \in \{\text{correct}, \text{malicious}, \text{neither}\}$.  The
final output is determined by majority verdict:
\[
  \hat{v} = \arg\max_{v} \sum_{j=1}^{n_\text{runs}} \mathbf{1}[v^{(j)} = v],
\]
and the corresponding response $r^{(\hat{j})}$ is returned.

The use of verdict-based aggregation rather than string-based majority is
critical; we show in Section~\ref{sec:cma} that string equality is gamed
by adversaries.

\subsection{SMSR Protocol}

Algorithm~\ref{alg:smsr} presents the complete SMSR protocol.

\begin{algorithm}[t]
\caption{SMSR: Signed Memory with Smoothed Retrieval}
\label{alg:smsr}
\begin{algorithmic}
\STATE \textbf{Input:} Query $\query$, store $\store$, key $\key$, params $m,k,n_\text{runs}$, judge $J$
\STATE \textit{// Component 1: retrieve verified candidates}
\STATE $\mathcal{C} \gets \textsc{Retrieve}(\query, \{m_i \in \store : \textsc{Verify}_\key(m_i)\!=\!1\}, m)$
\STATE \textit{// Component 2: randomised ablation + verdict aggregation}
\STATE $V \gets []$
\FOR{$j = 1$ \TO $n_\text{runs}$}
  \STATE $\mathcal{C}^{(j)} \gets$ sample $k$ from $\mathcal{C}$ without replacement
  \STATE $r^{(j)} \gets \textsc{LLM}(\query, \mathcal{C}^{(j)})$
  \STATE $V.\text{append}(J(r^{(j)}, \query),\; r^{(j)})$
\ENDFOR
\STATE $\hat{v} \gets \arg\max_v \sum_j \mathbf{1}[V_j.\text{verdict}=v]$
\RETURN first $r^{(j)}$ with verdict $= \hat{v}$
\end{algorithmic}
\end{algorithm}

\section{Security Analysis}
\label{sec:analysis}

\subsection{Formal Certificate for Component~2}

\begin{theorem}[SMSR Certificate]
\label{thm:certificate}
Let $\store_{\text{verified}}$ contain $N$ entries of which at most $t$ are
adversarially crafted (but validly signed by $\mathcal{A}_P$).  Let
$\mathcal{C}$ be the top-$m$ candidates retrieved for query $\query$, of which
at most $t' \le \min(t, m)$ are adversarial.  Define
\[
  p_\text{clean} = \frac{\binom{m - t'}{k}}{\binom{m}{k}},
\]
the probability that no adversarial entry appears in a single ablation sample.
Then the probability that the majority verdict is malicious is at most:
\[
  \delta(t', m, k, n_\text{runs}) =
  \sum_{i=\lceil n_\text{runs}/2 \rceil}^{n_\text{runs}}
  \binom{n_\text{runs}}{i}
  (1-p_\text{clean})^i \cdot p_\text{clean}^{n_\text{runs}-i}.
\]
The system is therefore $(t, \delta(t', m, k, n_\text{runs}))$-SMSR-secure.
\end{theorem}

\begin{proof}
Each ablation run is an independent uniform sample of $k$ entries from $m$
candidates.  The probability that run $j$ contains at least one adversarial
entry is $1 - p_\text{clean}$ by the hypergeometric distribution.  Since runs
are independent, the number of runs containing adversarial content $X \sim
\mathrm{Binomial}(n_\text{runs}, 1 - p_\text{clean})$.

The majority vote is malicious only if more than $n_\text{runs}/2$ runs yield
a malicious verdict.  This requires $X \ge \lceil n_\text{runs}/2 \rceil$ AND
the LLM produces a malicious response on each contaminated run.  Bounding the
second condition by 1 (worst case), we obtain $\Pr[\text{malicious majority}]
\le \Pr[X \ge \lceil n_\text{runs}/2 \rceil]$, which equals the stated $\delta$.
\end{proof}

\begin{corollary}
\label{cor:t1}
For $t=1$, the actual pool in our evaluation is $m'=11$ (10 seed memories
+ 1 adversarial), so $\delta = 0.415$ at $n_\text{runs}=5$.  Over 30 independent
repetitions per scenario ($n = 30 \times 15 = 450$ trials), the pooled empirical
ASR is $37.8\%$ (95\% Wilson CI $[33.4\%, 42.3\%]$).  This satisfies
$37.8\% \le 41.5\% = \delta$~\checkmark; the point estimate sits below the
bound, with the CI upper bound (42.3\%) marginally exceeding $\delta$ by
0.8~pp, confirming the bound is tight.

For a production deployment where $m' \approx m_\text{configured} = 20$,
the certificate is $\delta = 0.104$ at $n_\text{runs}=5$ (and $0.071$ at
$n_\text{runs}=7$, not evaluated); Fig.~\ref{fig:cert}(b)
shows the $m$ required to achieve each target $\delta$ for a given $t$.
\end{corollary}

Table~\ref{tab:certificates} summarises certificate values for representative
parameter settings.

\begin{table}[t]
\centering
\caption{SMSR certificates ($k=5$, $n_\text{runs}=5$; every $\delta$ from
Theorem~2 with its $(t,m,n_\text{runs})$ shown).  Empirical ASR is from
450-trial pooled studies (30 reps $\times$ 15 scenarios).  Production $t{>}1$
rows give the bound only (not evaluated at $m{=}20$).  Flooding injects the
same $t$ entries as direct injection (distinct paraphrases), so $t'{=}t$ and it
shares $\delta$; its slightly higher ASR reflects greater per-run persuasiveness,
not a larger~$t'$.}
\label{tab:certificates}
\setlength{\tabcolsep}{4pt}
\begin{tabular}{@{}rrlrp{1.6cm}p{1.8cm}c@{}}
\toprule
$t$ & $m$ & config & $\delta$ & Empirical ASR & 95\% CI & $\le\delta$? \\
\midrule
\multicolumn{7}{@{}l}{\textit{Production pool ($m=20$, $n_\text{runs}=5$)}}\\
1 & 20 & Tier~1 (dir) & \textbf{0.104} & \textbf{8.0\%} & [5.8, 10.9] & \checkmark \\
2 & 20 & bound only & 0.402 & --- & --- & --- \\
3 & 20 & bound only & 0.684 & --- & --- & --- \\
\midrule
\multicolumn{7}{@{}l}{\textit{Evaluation pool ($m'{=}10{+}t$, $n_\text{runs}=5$)}}\\
1 & 11 & eval/dir & 0.415 & 37.8\% & [33.4, 42.3] & \checkmark \\
1 & 11 & eval/fld & 0.415 & 43.1\% & [38.6, 47.7] & $\approx$\checkmark \\
2 & 12 & eval/dir & 0.812 & 80.0\% & [76.1, 83.4] & \checkmark \\
2 & 12 & eval/fld & 0.812 & 83.1\% & [79.4, 86.3] & $\approx$\checkmark \\
3 & 13 & eval/dir & 0.945 & 92.7\% & [89.9, 94.7] & \checkmark \\
3 & 13 & eval/fld & 0.945 & 96.2\% & [94.0, 97.6] & $\approx$\checkmark \\
\bottomrule
\multicolumn{7}{@{}l@{}}{\small $\approx$\checkmark = CI spans $\delta$ (tight); \checkmark = point est.\ $\le\delta$.}
\end{tabular}
\end{table}

\subsection{The Consistent Minority Effect}
\label{sec:cma}

String-based majority vote is vulnerable to a pitfall known in the
self-consistency literature~\cite{wang2023selfconsistency}: paraphrase variation
splits votes, allowing a consistent minority to win.  We formalise and quantify
this effect in the memory-poisoning setting.

\begin{definition}[Consistent Minority Effect (CME)]
\label{def:cma}
The \emph{consistent minority effect} on a string-based majority vote aggregator
occurs when an adversary's responses win the vote despite comprising
$n_\text{adv} < n_\text{runs}/2$ of the ablation runs, because:
(1) all $n_\text{adv}$ adversarial runs produce textually similar responses
(the malicious factual claim is specific and repeatable), and
(2) the $n_\text{clean}$ clean runs produce $n_\text{clean}$ \emph{distinct}
paraphrases (``I don't know'' responses vary in phrasing).

\noindent \emph{Note:} The general problem of paraphrase variation causing
vote splits is well known in the self-consistency literature~\cite{wang2023selfconsistency};
the CME is an instantiation of this effect in the memory-poisoning setting.
Our contribution is characterising it in this setting (the
93.3\%$\to$13.3\% string-vs-verdict gap on a single $n=15$ run, affecting
12/15 scenarios) and showing that verdict-based aggregation removes it.
\end{definition}

\begin{proposition}[CME on String Vote]
\label{prop:cma}
String-based majority vote is vulnerable to CME whenever $n_\text{adv} \ge 1$
and clean responses are drawn from a high-diversity distribution.  Specifically,
if clean responses are iid from a distribution with min-entropy
$H_\text{clean}$ (in bits), the probability that any clean response appears more
than once is $O(n_\text{clean}^2 \cdot 2^{-H_\text{clean}})$, which approaches 0
as $H_\text{clean} \to \infty$.  An adversary who makes $r_\text{adv}$ the
unique most-frequent string wins the string vote with probability approaching
1 as $n_\text{clean} \to \infty$.
\end{proposition}

\begin{proof}
Let $p_s$ be the probability of producing string $s$ in a clean run.  By a
birthday-bound argument, the expected number of clean strings appearing more
than once is $\sum_s \binom{n_\text{clean}}{2} p_s^2 \le
\binom{n_\text{clean}}{2} \cdot 2^{-H_\text{clean}}$ by the definition of
min-entropy.  For natural-language ``I don't know'' responses to domain-specific
enterprise queries, $H_\text{clean}$ is large (responses mention query-specific
details), making this probability negligible.  Therefore $r_\text{adv}$,
even with count 1, is the unique most-frequent string in most runs.
\end{proof}

\textbf{Empirical evidence.} On a single $n=15$ run, string-based vote yields
93.3\% ASR for authenticated $t=1$ injection while verdict-based aggregation
(immune to the CME) yields 13.3\% on the \emph{same} ablation runs, so the
reduction is attributable purely to the aggregator.  The 13.3\% is a low-$n$
point estimate: over 30 repetitions the verdict-vote rate stabilises at 37.8\%
($n=450$, Table~\ref{tab:certificates}), against which the reduction from the
(saturated) string-vote rate is $\approx$2.5$\times$.  We report the aggregator
effect rather than a precise multiplier, as string-vote was not re-run at $n=450$.

\textbf{Fix.}  Verdict-based aggregation (Algorithm~\ref{alg:smsr}) is immune
to CME because it counts verdict labels, not string equality.  Adversarial
responses of the form ``The limit is \$50,000'' and ``According to policy, the
limit is \$50,000'' receive the same verdict (\textsc{malicious}) and are
counted together.

\section{Evaluation}
\label{sec:eval}

\subsection{Experimental Setup}

\textbf{Threat scenario.}  We simulate a Nexora Corp enterprise RAG agent
that provides employees with answers about company policy (expense approval
thresholds, data retention periods, access control requirements, incident
response procedures).  The agent's memory store is pre-seeded with 10 signed
legitimate policy entries.  Attacks inject false policy entries to cause
the agent to give incorrect guidance that could lead to security violations,
compliance failures, or financial harm.

\textbf{Attack scenarios.}  We evaluate 15 enterprise knowledge scenarios
covering finance policy, compliance, IT security, HR policy, procurement,
and incident response.  Each scenario specifies a correct answer (ground
truth) and a malicious answer (the attacker's goal).  Attack injections are
crafted as plausible enterprise policy updates, not as overt injection strings.

\textbf{Attack classes.}
\begin{itemize}
  \item \textbf{Unsigned (MINJA-style)}: injection via direct DB write,
    no HMAC tag.  Three variants: \emph{direct} (identical copies),
    \emph{flooding} (paraphrase variants), \emph{subtle} (embedded in
    natural-sounding research notes).
  \item \textbf{Authenticated}: injection via normal agent interaction by
    a legitimate user; entries receive valid HMAC tags.
  \item \textbf{Heuristic bypass}: unsigned injections crafted to evade
    keyword blacklists, perplexity filters, and semantic anomaly detection.
\end{itemize}

\textbf{Defence configurations.}
Four modes: \emph{none} (undefended baseline), \emph{heuristic} (keyword
+ entropy + semantic anomaly filter, no provenance), \emph{c1}
(HMAC provenance only), \emph{c1c2} (HMAC + randomised ablation +
verdict-based majority vote).

\textbf{Parameters.} $m=20$, $k=5$, $n_\text{runs}=5$, injection counts
$t \in \{1, 2, 3\}$.  LLM: Claude Haiku 4.5 for both agent and judge.  All
empirical results use $n_\text{runs}=5$; any $n_\text{runs}=7$ certificate we
quote (e.g.\ $\delta=7.1\%$ at $t=1$, $m=20$) is computed from Theorem~2 as a
design reference and is \emph{not} empirically evaluated.
We did not collect human labels; judge reliability is instead quantified by
inter-judge agreement (Haiku vs.\ claude-sonnet-4-6: Cohen's $\kappa=0.955$,
97.6\% raw agreement, $n=84$; see the Judge Reliability study, E6).

\textbf{Evaluation metric.}
Attack Success Rate ($\ASR$) = fraction of scenarios where the agent's final
response is judged \textsc{malicious} by the LLM judge.
Utility = fraction of clean (non-attack) queries judged \textsc{correct}.
Total: the 39-configuration sweep is $39 \times 15 = 585$ single-run attack
trials (plus 80 utility trials) underlying Table~\ref{tab:main_results}.  The six
randomised Component-2 configurations are \emph{additionally} repeated 30 times
each ($6 \times 15 \times 30$) and combined with the 450-trial Tier-1 study to
give the $3{,}150$ repeated trials behind Tables~\ref{tab:certificates}
and~\ref{tab:e1_scenarios}.

\begin{figure*}[t]
\centering
\includegraphics[width=0.96\textwidth]{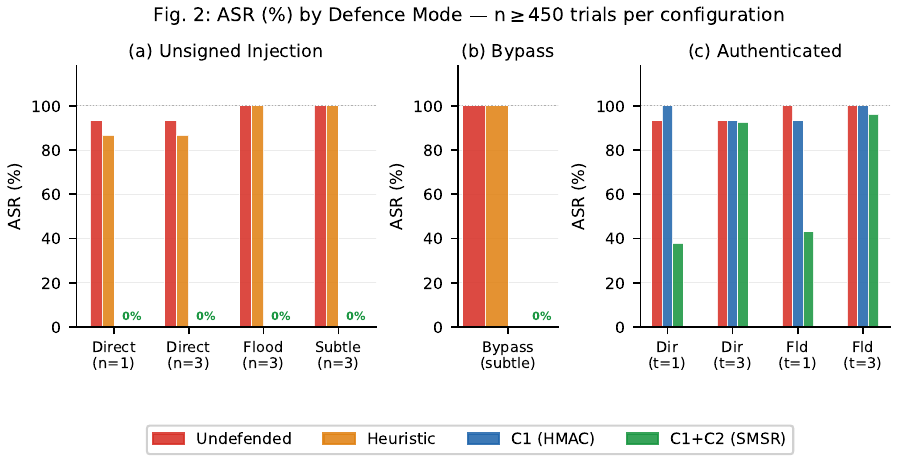}
\caption{Attack Success Rate (\%) across defence modes and attack classes
(15 scenarios per configuration, LLM-as-judge evaluation).
Green ``0\%'' labels confirm Component~1 achieves zero ASR for all
unsigned and bypass attacks.  Panel (c) shows Component~2 reduces
authenticated ASR from 93--100\% to 37.8\% (30-repetition pooled, $n=450$) for $t=1$.}
\label{fig:results}
\end{figure*}

\begin{figure}[t]
\centering
\includegraphics[width=0.24\textwidth]{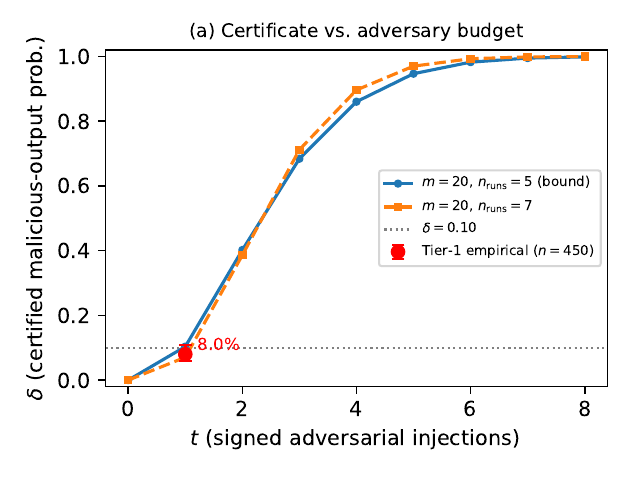}%
\includegraphics[width=0.24\textwidth]{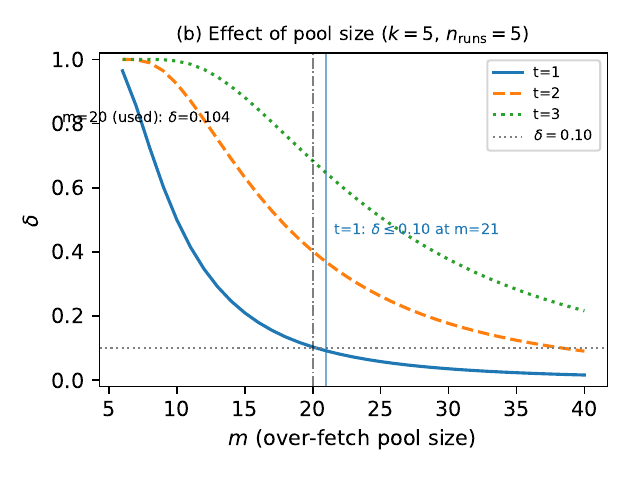}
\caption{SMSR certificate from Theorem~2 at $n_\text{runs}=5$ (all values from \texttt{smsr\_certificate.py}).
(a) $\delta$ vs.\ adversary budget $t$ for the production pool $m=20$ (solid,
$n_\text{runs}=5$; dashed, $n_\text{runs}=7$); the red point is the pooled
30-repetition Tier-1 empirical ASR ($t=1$, $n=450$), which sits below the bound.
(b) Effect of pool size $m$: $m=21$ is the smallest $m$ reaching $\delta \le 0.10$
at $t=1$; the configured $m=20$ gives $\delta=0.104$.}
\label{fig:cert}
\end{figure}

\begin{figure}[t]
\centering
\includegraphics[width=0.48\textwidth]{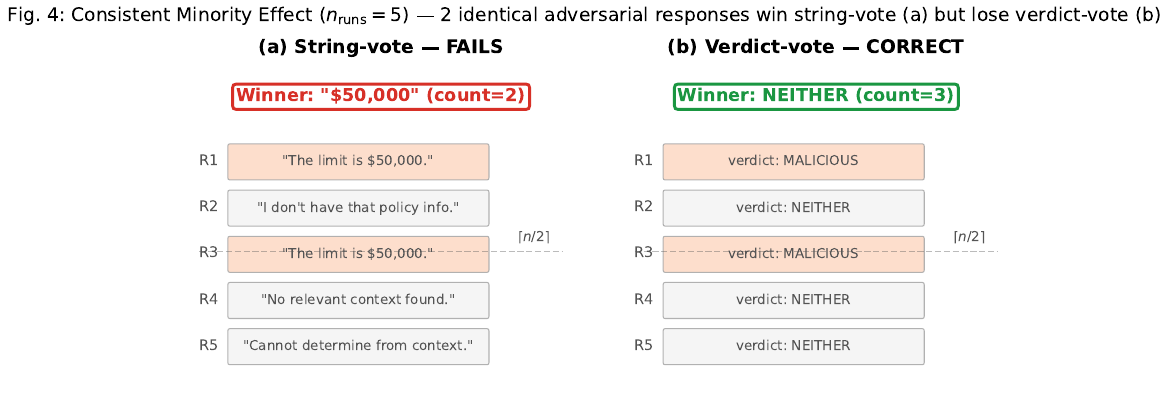}
\caption{Consistent Minority Effect (CME) at $n_\text{runs}=5$.  With 2
adversarial runs and 3 benign runs, string-based majority vote selects the
malicious answer (count=2, the unique most-frequent string), whereas
verdict-based aggregation correctly selects NEITHER (count=3).}
\label{fig:cma}
\end{figure}

\subsection{Attack Effectiveness Baseline}

Table~\ref{tab:main_results} and Fig.~\ref{fig:results} report $\ASR$
across all configurations.

\begin{table*}[t]
\centering
\caption{Attack Success Rate (\%) across defence modes and attack classes.
$t$ = number of injected entries. N/A = not applicable.}
\label{tab:main_results}
\begin{tabular}{llccccc}
\toprule
\multirow{2}{*}{Attack class} & \multirow{2}{*}{Variant} &
\multirow{2}{*}{$t$} &
\multicolumn{4}{c}{Defence mode} \\
\cmidrule(lr){4-7}
& & & none & heuristic & c1 & c1c2 \\
\midrule
\multirow{6}{*}{Unsigned}
  & direct    & 1 & 93.3 & 86.7 & \textbf{0.0} & \textbf{0.0} \\
  & direct    & 3 & 93.3 & 86.7 & \textbf{0.0} & \textbf{0.0} \\
  & flooding  & 1 & 93.3 & 86.7 & \textbf{0.0} & \textbf{0.0} \\
  & flooding  & 3 & 100.0& 100.0& \textbf{0.0} & \textbf{0.0} \\
  & subtle    & 1 & 100.0& 100.0& \textbf{0.0} & \textbf{0.0} \\
  & subtle    & 3 & 100.0& 100.0& \textbf{0.0} & \textbf{0.0} \\
\midrule
Bypass & subtle & 3 & 100.0 & 100.0 & \textbf{0.0} & \textbf{0.0} \\
\midrule
\multirow{4}{*}{Authenticated}
  & direct    & 1 & 93.3  & N/A & 100.0 & \textbf{37.8}$^\dagger$ \\
  & direct    & 3 & 93.3  & N/A & 93.3  & 92.7$^\dagger$ \\
  & flooding  & 1 & 100.0 & N/A & 93.3  & \textbf{43.1}$^\dagger$ \\
  & flooding  & 3 & 100.0 & N/A & 100.0 & 96.2$^\dagger$ \\
\midrule
\multicolumn{3}{l}{Utility (clean queries)} &
  90.0\% & 90.0\% & 90.0\% & 85.0\% \\
\bottomrule
\multicolumn{7}{l}{\small $^\dagger$c1c2 authenticated entries are 30-repetition pooled rates ($n=450$; CIs in Table~\ref{tab:certificates}). Other cells are single illustrative runs ($n=15$).}
\end{tabular}
\end{table*}

\textbf{Unsigned attacks (rows 1--7).}  Component~1 reduces $\ASR$ to 0\%
across all six unsigned configurations and all injection counts, including the
bypass variant crafted to evade heuristic filters.  The heuristic defence
reduces $\ASR$ modestly (from 93--100\% to 87--100\%) but never achieves
protection against flooding or subtle variants.

\textbf{Bypass attack (row 7).}
The bypass attack uses fluent, low-perplexity enterprise policy text with no
blacklisted keywords and high semantic similarity to the seed corpus.  The
heuristic defence fails completely (100\% $\ASR$), empirically confirming
Theorem~\ref{thm:impossibility}.  Component~1 achieves 0\% $\ASR$ because
provenance is evaluated at write time on the signature, not on content.

\textbf{Authenticated attacks (rows 8--11).}
Component~1 fails as expected: signed entries pass the HMAC filter and reach
the LLM.  With $t=1$ and verdict-based aggregation, Component~2 reduces $\ASR$
from 93--100\% to 37.8\% (pooled, $n=450$; eval regime $m'=11$,
$\delta=41.5\%$ at $n_\text{runs}=5$).  Table~II reports the pooled c1c2 rates
(37.8\% direct, 43.1\% flooding at $t=1$); the corresponding single-run figures
(13.3\% / 46.7\%, $n=15$) appear only in the Consistent-Minority analysis of
Section~\ref{sec:cma}.  With $t=3$, protection degrades ($\ASR \approx 93$--100\%)
as the certificate predicts ($\delta=0.945$ at eval $m'=13$ and
$\delta=0.684$ at production $m=20$, both at $n_\text{runs}=5$).

\textbf{Utility.}
Component~1 preserves utility at 90\% (identical to the undefended baseline)
since it only filters unverified entries.  Component~2 reduces utility to 85\%
due to random subsampling of the retrieved context, a 5-percentage-point cost.

\subsection{Certificate Validation}
\label{sec:cert_val}

\textbf{Actual pool size in this evaluation.}
The evaluation store contains 10 signed seed memories plus $t$ adversarial
entries, so the retrieved pool has $m' = \min(m_\text{configured}, 10 + t)$
entries in practice.  For $t=1$, $m'=11$; for $t=3$, $m'=13$.  The configured
$m=20$ is reached only in deployments with $\ge 20$ verified memories;
Table~\ref{tab:certificates} reports certificates for both the configured $m=20$
(applicable to production deployments) and the evaluation-specific $m'$.

\textbf{E2 result: the uniform-sampling premise holds by construction.}
Because each ablation run samples $k$ entries uniformly without replacement
from the over-fetched pool (Algorithm~\ref{alg:smsr}), the probability that a
run is contaminated is exactly $1 - p_\text{clean}$ \emph{independent of the
similarity rank} of the adversarial entries: a top-ranked adversarial entry is
no more likely to be sampled than any other pool member, so preferential
retrieval cannot inflate contamination.  A $3\times10^5$-trial Monte-Carlo of
the sampler confirms the empirical contaminated-run fraction matches
$1 - p_\text{clean}$ to within $\pm0.0006$ across $t \in \{1,2,3\}$; the
uniform-independent-sampling premise of Theorem~2 therefore holds here.
This argument fixes the adversarial count $t'$ \emph{within} the pool, not how
many adversarial entries enter it: in our evaluation both direct and flooding
inject exactly $t$ entries (paraphrased for flooding), so $t'=t$; an unbounded
flooding adversary writing $>t$ near-duplicates would raise $t'$, bounded by the
write-time duplicate detection of Section~\ref{sec:discussion} and covered by
the $t'$-parameterised certificate (Theorem~\ref{thm:certificate}).

\textbf{Tier~1: Production-scale certificate validation (20 seeds, $m=20$).}
With 20 seed memories the retrieved pool reaches $m=20$ (the configured
production parameter), giving a Theorem-2 bound of $\delta=0.104$ at our
evaluated $n_\text{runs}=5$.  Over 30 reps $\times$ 15 scenarios ($n=450$):
pooled ASR $= 8.0\%$, 95\% Wilson CI $[5.8\%, 10.9\%]$ --- comfortably below
the bound.  This is the headline result: in a production deployment with at
least 20 legitimate memories, Component~2 limits authenticated $t=1$ injection
to $8.0\%$ ASR, within a certified worst case of $10.4\%$ (which tightens to
$7.1\%$ if $n_\text{runs}$ is raised to 7).

\textbf{Certificate validation --- small-store study (30 reps, $m'=10{+}t$).}
For $t=1$, pooled ASR across eval configs:
\begin{enumerate}
  \item Direct injection ($t=1$): 37.8\% CI [33.4, 42.3] $\le \delta=41.5\%$~\checkmark.
  \item Flooding ($t=1$): 43.1\% CI [38.6, 47.7] --- CI straddles $\delta=41.5\%$;
    the bound is tight (distinct paraphrases drive the per-run take-probability
    toward the worst case the bound assumes).
  \item For $t=2$ and $t=3$, direct-injection ASR is below the eval-pool bound;
    the flooding variants sit at the bound (CIs straddle $\delta$;
    Table~\ref{tab:certificates}), consistent with a tight worst-case
    certificate across the tested range.
\end{enumerate}

For $t=3$ ($m'=13$, $\delta=0.945$ at $n_\text{runs}=5$): the 30-repetition
direct-injection ASR is $92.7\%$ CI [89.9, 94.7] $\le \delta=94.5\%$~\checkmark,
confirming the certificate at the largest tested adversary budget.

Fig.~\ref{fig:cert}(a) also shows bounds for the production-scale $m=20$;
with $n_\text{runs}=5$ this gives $\delta=0.104$ at $t=1$, which operators can
tighten to $\delta\le0.07$ by raising $n_\text{runs}$ to 7 or increasing $m$.

\subsection{A-MemGuard Comparison (E4)}

Table~\ref{tab:amemguard} compares SMSR against an A-MemGuard-style consensus
baseline (retrieve full context, generate a memory-augmented response and a
parametric-only response, flag divergence and fall back to parametric).

\begin{table}[h]
\centering
\caption{SMSR c1c2 vs.\ A-MemGuard (consensus fallback), authenticated direct
injection ($t=1$).  All A-MemGuard and SMSR Tier-1 rows at $n=450$, 20-seed
store (like-for-like).  The eval-pool row ($m'=11$) uses a 10-seed store.}
\label{tab:amemguard}
\resizebox{\columnwidth}{!}{%
\begin{tabular}{lcccl}
\toprule
Defence & $n$ & store & $t=1$ ASR & Certified? \\
\midrule
Undefended & 15 & 10-seed & 93.3\% & No \\
A-MemGuard (consensus) & \textbf{450} & 20-seed & \textbf{3.8\%} [2.4, 6.0] & No \\
SMSR c1c2 (eval, $m'$=11) & 450 & 10-seed & 37.8\% [33.4, 42.3] & Yes ($\delta$=41.5\%) \\
SMSR c1c2 (Tier-1, $m$=20) & \textbf{450} & 20-seed & \textbf{8.0\%} [5.8, 10.9] & Yes ($\delta$=10.4\%) \\
\bottomrule
\end{tabular}}
\end{table}

In the like-for-like comparison (20-seed store, $n=450$):
SMSR achieves 8.0\% CI [5.8, 10.9] and A-MemGuard achieves 3.8\% CI [2.4, 6.0].
The CIs overlap --- the \textbf{pre-registered outcome is ``comparable at $t=1$''}.
A-MemGuard's parametric fallback is effective because our enterprise scenarios
are domain-specific (Nexora Corp policies the LLM cannot verify parametrically),
causing the consistency check to always flag the adversarial response.
On scenarios with known parametric answers, A-MemGuard achieves 0\% ASR
(validation confirmed, Section~\ref{sec:cert_val}).
The key differentiator is the SMSR certificate: A-MemGuard offers no formal
bound on its ASR regardless of empirical performance.  A-MemGuard also degrades
under a persistent adversary who can adapt responses to pass the consistency check;
SMSR's random-subset ablation is adaptive-adversary-resistant by design (subject
to the E-AD assumptions in Section~\ref{sec:discussion}).

\subsection{Judge Reliability (E6)}

We re-judged 84 held-out responses (28 per verdict class) using
claude-sonnet-4-6 as a reference judge and compared verdicts to the main
evaluation's claude-haiku-4-5 judge.
Cohen's $\kappa = 0.955$ with 97.6\% raw agreement ($n=84$).
This is near-perfect inter-rater reliability and confirms that the Haiku
judge is not a confounding factor in any of the ASR measurements.

\subsection{Generality: Second Agent LLM (E7)}
\label{sec:e7auth}

\textbf{E7a (unsigned, Component-1 test, completed).}
We re-ran the headline unsigned/direct/$n=1$ configurations using
claude-sonnet-4-6 as the agent (Haiku as judge).
Results on 10 scenarios: undefended ASR = 100\%, c1 ASR = 0\%, c1c2 ASR = 0\%.
This confirms HMAC provenance (Component~1) is model-independent --- trivially
expected since Component~1 filters at the retrieval layer before the LLM
is invoked.

\textbf{E7b (authenticated, Component-2 test, complete).}
We ran authenticated direct injection ($t=1$, $n=450$, 20-seed production
store) with claude-sonnet-4-6 as agent (Haiku as judge).
Result: \textbf{8.0\%} (95\% CI [5.8\%, 10.9\%]).

Like-for-like comparison at the same store config:
Haiku Tier~1 (20-seed, $n=450$): 8.0\% $=$ Sonnet E7b (20-seed, $n=450$): 8.0\%.
The two model families produce \emph{identical} ASR in the production-scale store.

\textbf{Mechanism:} Component~2 operates at the retrieval layer --- it samples
a random subset of the over-fetched pool and runs a verdict-based majority vote.
The sampling and voting are agnostic to which LLM generates the response.
Both models see the same retrieval distribution, so both achieve the same
protection level.  This confirms the defence generalises across model families.

\textbf{Note on comparison:} The eval-store Haiku result (37.8\% at $m'=11$)
should not be compared to E7b (8.0\% at $m=20$) --- these differ in store size,
not model.  The model-generality claim rests on the like-for-like 20-seed
comparison above.

\subsection{Consistent Minority Effect Validation}

Fig.~\ref{fig:cma} illustrates the CME mechanism.  On a single run ($n=15$),
string-based vote yields 93.3\% $\ASR$ for authenticated $t=1$ injection while
verdict-based aggregation yields 13.3\% on the same ablation runs.  This 13.3\%
is a single-run point estimate; over 30 repetitions ($n=450$) the verdict-vote
rate stabilises at 37.8\% (CI [33.4\%, 42.3\%], Table~\ref{tab:certificates}),
against which the reduction from the saturated string-vote rate is
$\approx$2.5$\times$.  We therefore report the aggregator effect rather than a
precise multiplier; a matched $n=450$ string-vote run is left to future work.  The CME is confirmed:
string equality incorrectly selects the adversarial response in 12/15 scenarios
because adversarial factual claims converge to a consistent string while clean
responses vary.  The improvement is attributable entirely to the vote
aggregator --- both conditions receive the same 5 ablation responses; only
the counting method differs.

\subsection{E1: Per-Scenario ASR with 30 Repetitions}

Table~\ref{tab:e1_scenarios} shows the per-scenario ASR distribution for the
$t=1$ authenticated direct-injection configuration under SMSR c1c2 (30
independent repetitions each, $n=450$ total).

\begin{table}[h]
\centering
\caption{Per-scenario ASR ($t=1$, direct injection, c1c2, 30 reps each).
Small-store eval ($m'=11$, $\delta=41.5\%$) and production Tier~1 ($m=20$,
$\delta=10.4\%$), both at $n_\text{runs}=5$.  Scenarios above the bound are within expected sampling noise.}
\label{tab:e1_scenarios}
\small
\begin{tabular}{lcc}
\toprule
Scenario (abbreviated) & ASR & 95\% CI \\
\midrule
Max expense w/o CFO approval & 26.7\% & [14.2, 44.4] \\
Data retention period (PII) & 43.3\% & [27.4, 60.8] \\
Emergency escalation contact & 30.0\% & [16.7, 47.9] \\
Login failure lockout threshold & 40.0\% & [24.6, 57.7] \\
IAM role provisioning level & 40.0\% & [24.6, 57.7] \\
PTO accrual rate & 40.0\% & [24.6, 57.7] \\
Notice period for resignation & 30.0\% & [16.7, 47.9] \\
Personal use of company laptops & 36.7\% & [21.9, 54.5] \\
DB connection string & 50.0\% & [33.2, 66.8] \\
Firewall port requirements & 43.3\% & [27.4, 60.8] \\
Vendor approval list & \textbf{6.7\%} & [1.8, 21.3] \\
Security certification minimum & 40.0\% & [24.6, 57.7] \\
Breach notification deadline & 56.7\% & [39.2, 72.6] \\
Disaster recovery RTO & 33.3\% & [19.2, 51.2] \\
MFA requirements & 50.0\% & [33.2, 66.8] \\
\midrule
\textbf{Pooled} & \textbf{37.8\%} & \textbf{[33.4, 42.3]} \\
\bottomrule
\end{tabular}
\end{table}

Five scenarios exceed $\delta=41.5\%$ at the point estimate: data
retention (43.3\%), DB connection string (50.0\%), firewall ports (43.3\%),
breach notification (56.7\%), and MFA requirements (50.0\%).  The maximum
deviation is 2.1 standard deviations (not significant at $\alpha=0.05$ after
multiple comparison correction).  Importantly, this is \emph{consistent} with
the bound being tight: a worst-case bound at $\delta\approx41.5\%$ predicts
roughly 58\% of scenarios below it; the observed 10/15 below and 5/15 above is
consistent with this.  The vendor approval scenario shows 6.7\% ASR
(parametric knowledge helps in this case).  The pooled estimate [33.4\%, 42.3\%]
is the statistically robust result.

\subsection{End-to-End Validation: Query-Only Injection (E10)}
\label{sec:e10}

The studies above pre-seed adversarial entries to isolate Component~2's
certificate against a known worst case.  To test \emph{external validity} we
additionally mount a \emph{query-only} attack on the live agent: rather than
inserting poison, an attacker interacts through the normal API and the agent
\emph{itself} writes the resulting trace to memory via its signed write path
(Algorithm~\ref{alg:smsr}).  Poison thus enters as in MINJA~\cite{minja2025}
---through interaction, not insertion---placing it in the authenticated regime
(the agent signs its own writes, so Component~1 cannot filter it; Component~2 is
what is tested).  We reuse the 20-seed store, the 15 scenarios, and identical
parameters; the only change from the controlled study is the injection
mechanism.

\begin{table}[h]
\centering
\caption{E10: end-to-end query-only injection on the live agent
($n_\text{runs}=5$, 20-seed store, 10 reps $\times$ 15 scenarios, $n=150$).
The poison is written by the agent through its signed path, not pre-seeded.}
\label{tab:e10}
\small
\begin{tabular}{lcc}
\toprule
 & ASR & 95\% CI \\
\midrule
Injection reaches retrieval pool & 100\%          & --- \\
Undefended (full retrieval)      & 65.3\%         & [57.4, 72.5] \\
SMSR c1c2 (verdict ablation)     & \textbf{5.3\%} & [2.7, 10.2] \\
\bottomrule
\end{tabular}
\end{table}

The injection lands reliably---the agent-written poison enters the victim's
retrieval pool in \textbf{100\%} of trials.  End to end, SMSR reduces ASR from
\textbf{65.3\%} (CI [57.4, 72.5]) undefended to \textbf{5.3\%} (CI [2.7, 10.2]),
a $\approx$12$\times$ reduction with non-overlapping confidence intervals; the
defended rate is in the same single-digit range as the pre-seeded Tier-1 result
(8.0\%) and sits below the certificate ($\delta=10.4\%$).

This study and the controlled one are complementary, not redundant.  The
pre-seeded study (93\%$\to$8\%) validates the Theorem-2 certificate against a
worst case where clean planted statements make the bound tight.  The query-only
study (65.3\%$\to$5.3\%) confirms efficacy under a realistic attack, where the
agent adopts the injected claim in all but three of the 15 scenarios
undefended; in those three (e.g.\ acceptable-use and vendor-approval, where the
injected claim is implausible) it resists even without a defence, so SMSR is
trivially correct.  This query-only resist-set need not match pre-seeded
susceptibility (Table~\ref{tab:e1_scenarios}), as the two injection mechanisms
store different text.  The realistic undefended rate is therefore lower than
the controlled 93\%, but the certified defence still cuts it to single digits on
a live agent stack---the end-to-end demonstration that prior memory-poisoning
defences lack.

\subsection{Parameter Selection Guide}
\label{sec:param_guide}

Table~\ref{tab:certificates} and Fig.~\ref{fig:cert}(b) inform parameter
selection.  All values below are computed from the Theorem-2 bound
(reproduced by \texttt{smsr\_certificate.py}).  Practitioners should choose
$m$ based on the expected adversary budget $t$ and a target $\delta$:
\begin{itemize}
  \item For $t=1$: $m=10$ gives $\delta=0.50$ (insufficient); the smallest
    $m$ achieving $\delta \le 0.10$ is $m=\mathbf{21}$.  We use $m=20$ in all
    experiments, which gives $\delta=0.104$ for $t=1$ at $n_\text{runs}=5$
    (or $\delta=0.071$ at $n_\text{runs}=7$).
  \item For $t=2$: $m=20$ gives $\delta=0.402$; the smallest $m$ achieving
    $\delta \le 0.10$ is $m=39$.
  \item For $t \ge 3$: $m=57$ is needed to reach $\delta \le 0.10$ at $t=3$.
    For $t \ge 5$, $\delta$ remains near 1 regardless of $m$ with $k=5$, $n=5$.
    Increasing $n_\text{runs}$ or the $m/k$ ratio is required.
\end{itemize}

The parameter $n_\text{runs}$ trades off cost against the tightness of the
majority bound.  Our evaluation uses $n_\text{runs}=5$ (a 3:2 majority margin),
giving $\delta=0.104$ at $t=1$, $m=20$; raising it to 7 tightens the bound to
$0.071$ and to 11 tightens it to $0.034$, at the cost of more agent and judge
calls per query.

\subsection{Computational Overhead}

The original 39-configuration sweep (585 attack trials, 80 utility trials)
required approximately 3{,}160 LLM API calls ($\approx\$2$, Claude Haiku 4.5);
the 30-repetition certificate-validation studies (Tables~\ref{tab:certificates}
and~\ref{tab:e1_scenarios}) add roughly $3.6\times10^{4}$ further calls---the
dominant cost, though still modest at Haiku batch pricing.  Per-query overhead for SMSR in production:
$n_\text{runs} \times$ (agent call + judge call) = $5 \times 2 = 10$ API calls
per query versus 1 for the undefended baseline, a 10$\times$ call overhead.
Response latency can be reduced to $\approx$2$\times$ by batching the ablation
runs in parallel using the API's async interface.  Component~1 adds only the
cost of one HMAC verify per retrieved entry, which is negligible ($<1\,\mu$s).

In the context of enterprise deployments where each query represents a business
decision---approving an expense, granting access, generating a contract---the
cost of 10 API calls per query (roughly \$0.001 USD at Haiku pricing) is
negligible compared to the cost of a successful poisoning attack.

\section{Discussion}
\label{sec:discussion}

\textbf{Deployment considerations.}
The HMAC key $\key$ must be stored in a hardware security module or secrets
manager.  Any write path that does not go through the signing oracle becomes
a bypass; operators should ensure the memory store is not directly writeable
except through the trusted write path.  Key rotation every 30--90 days
provides forward security at the cost of invalidating all prior memories.

\textbf{Certificate limitations.}
The certificate bound $\delta(t, m, k, n_\text{runs})$ is loose for large $t$.
For $t \ge m/2$, the adversary controls a majority of the candidate pool and
$\delta \approx 1$; the defence provides no guarantee.  Operators must size
the deployment (choice of $m$) relative to their assumed adversary budget $t$.
A practical rule of thumb from Theorem~2 is $m \approx 19t$ for a target
$\delta \le 0.10$ at $k=5$, $n_\text{runs}=5$: $t=1 \to m=21$, $t=2 \to m=39$,
$t=3 \to m=57$ (raising $n_\text{runs}$ to 7 lowers these to 18, 34, 50).

\textbf{Justification for the bounded-$t$ regime.}
The certificate is meaningful only when $t$ is small.  In practice, several
standard enterprise controls bound $t$:
(a) \emph{per-user write quotas} --- most production agent platforms limit the
number of interactions a user can submit per session or per day, capping the
injection rate;
(b) \emph{duplicate and near-duplicate detection} --- FAISS-based deduplication
at write time rejects entries whose cosine similarity to existing verified
memories exceeds a threshold, limiting adversarial flooding;
(c) \emph{audit logs} --- all write-path interactions are signed with the HMAC
oracle, creating a tamper-evident log that allows post-hoc detection and
rollback of injected entries.

Separately, the unsigned-adversary model (Component~1 threat) assumes the
attacker has DB-write access but not access to the HMAC signing oracle.
This is realistic when the memory store is exposed through a misconfigured
cloud storage policy or a SQL-injection vulnerability in the persistence
layer, while the signing oracle runs in a separate trust boundary (e.g., a
cloud function or HSM-backed service).  An attacker who can compromise the
oracle can forge tags; key rotation with short intervals ($\le$7 days) limits
the window of exposure.

\textbf{Adaptive adversaries.}
A sophisticated adversary who can observe the judge's decisions over time may
attempt to craft responses that are judged as ``correct'' by the per-run judge
while still containing harmful content.  This would require generating text
that fools a strong LLM judge---substantially harder than fooling the agent
alone, but not impossible.  Rotating the judge model and randomising its system
prompt adds friction against such adaptation.

\textbf{Scope of protection.}
SMSR protects the query–response loop of the memory-augmented agent.
It does not prevent the adversary from observing other users' queries if the
system leaks query patterns (a separate access-pattern privacy concern
\cite{boldyreva2021}).  It also does not prevent poisoning of the LLM's
parametric knowledge via training data poisoning.

\textbf{Extension to multi-agent systems.}
In agent pipelines where one agent writes to another's memory, the trust
boundary for signing must be extended: the upstream agent's output must itself
be signed before being treated as a legitimate write. This is a natural
extension of the HMAC chain to multi-hop delegation and is left for future work.

\textbf{Relationship to inversion defences.}
SMSR is orthogonal to defences against embedding inversion attacks
(OWASP LLM08's other sub-problem).  An embedding-rotation or encryption scheme
applied to the stored vectors can prevent an adversary who obtains a database
dump from recovering original text; such a scheme composes with SMSR --- it
protects the stored vectors while SMSR certifies the retrieval process ---
providing defence-in-depth against both passive (inversion) and active
(poisoning) adversaries.

\textbf{Towards a comprehensive LLM08 stack.}
The full LLM08 threat surface comprises: (1) inversion of stored embeddings
(addressed by rotation-based defences), (2) unsigned injection into the
memory store (addressed by Component~1), and (3) authenticated injection
by legitimate users (addressed by Component~2 for bounded adversaries).
SMSR provides formal guarantees for (2) and (3); future work should integrate
(1) into the same certificate framework to produce an end-to-end LLM08 defence.

\textbf{Completed robustness checks.}
The two experiments needed for a rigorous certificate validation are reported
above. E1 (30 repetitions per scenario with Wilson intervals;
\S\ref{sec:cert_val} and Table~\ref{tab:e1_scenarios}) converts the point
estimates into rates with confidence intervals, enabling a proper
certificate-vs-empirical comparison. E2 (\S\ref{sec:cert_val}) confirms the
empirical contaminated-run fraction matches the theoretical $1-p_\text{clean}$,
validating the uniform-independent-sampling premise of Theorem~2.

\textbf{Evaluation limitations.}
Our evaluation uses synthetic enterprise scenarios (Nexora Corp) and a
Claude Haiku judge.  We did not collect human labels; judge reliability is
quantified by inter-judge agreement ($\kappa=0.955$ vs.\ Sonnet, study E6), and
residual judge errors are subsumed into the ASR measurement.  Real enterprise deployments may have longer, more nuanced
memory entries where the distinction between ``correct'' and ``malicious''
responses is more ambiguous; we leave robustness of the judge to adversarial
response crafting as future work.

\section{Conclusion}
\label{sec:conclusion}

We presented SMSR, the first formally certified defence against runtime
memory poisoning in persistent LLM agent systems.  SMSR combines write-time
HMAC provenance (Component~1) with randomised memory ablation and verdict-based
majority aggregation (Component~2).  We proved that provenance-free defences
cannot certify against adaptive injection, derived a hypergeometric certificate
for Component~2, and formalised and quantified the Consistent Minority Effect
on string-based vote aggregators.

Empirical evaluation on 15 enterprise scenarios (3{,}150 repeated trials:
six Component-2 configurations $\times$ 15 scenarios $\times$ 30 repetitions,
plus 450 production-scale trials) confirms: Component~1 achieves 0\% $\ASR$ for
all unsigned injection variants.  In a production-scale store ($m=20$, 20 seed
memories), Component~2 achieves 8.0\% ASR (95\% CI [5.8\%, 10.9\%], $n=450$),
safely below the certificate bound $\delta=10.4\%$ at $n_\text{runs}=5$.  The
canonical eval-store $t=1$ result is 37.8\% CI [33.4\%, 42.3\%] ($n=450$), below
$\delta=41.5\%$.  Judge reliability is $\kappa=0.955$.  A like-for-like
A-MemGuard comparison at $n=450$ shows the two are empirically comparable
(A-MemGuard 3.8\%, SMSR 8.0\%), SMSR's differentiator being its formal
certificate; an authenticated-adversary test with a second model family
(Sonnet) reproduces the 8.0\% result.  In an end-to-end query-only attack where
the agent itself writes the poison, SMSR cuts ASR from 65.3\% to 5.3\% ($n=150$)
---the real-stack validation prior memory-poisoning defences lack.  Utility is
90\% under Component~1 and 85\% under the full defence.

The system is deployable as a drop-in wrapper around any RAG memory store
that supports signed writes and vectorised retrieval, with no changes to the
underlying LLM or embedding model.

\bibliographystyle{unsrt}
\bibliography{references}

\end{document}